\pdfoutput=1

\documentclass[11pt, twoside, a4paper]{article}
\usepackage[top=1.25in, bottom=1.25in, left=1.25in, right=1.25in, a4paper]{geometry}
\usepackage[T1]{fontenc}
\usepackage{lmodern}
\usepackage{amsthm}
\usepackage{amsfonts}
\usepackage{amssymb}
\usepackage{amsmath}
\usepackage{mathtools}
\usepackage{float}
\usepackage{graphicx}
\usepackage{caption}
\usepackage{subcaption}
\usepackage[newenum]{paralist}
\usepackage{verbatim}
\usepackage[splitrule, bottom]{footmisc}
\usepackage{enumitem}
\usepackage{url}
\usepackage{hyperref}
\usepackage[noend]{algpseudocode}
\usepackage{algorithm}
\renewcommand{\algorithmicrequire}{\textbf{Input:}}
\renewcommand{\algorithmicensure}{\textbf{Output:}}

\makeatletter
\usepackage[usenames,dvipsnames,svgnames,table]{xcolor}
\usepackage{varwidth}
\usepackage{calc} 

\newcommand{\setensurelength}{%
    \settowidth{\@tempdima}{\algorithmicensure}%
    \setlength{\@tempdima}{\dimexpr\linewidth-\@tempdima-1em+\@totalleftmargin}}
\newcommand{\setrequirelength}{%
    \settowidth{\@tempdima}{\algorithmicrequire}%
    \setlength{\@tempdima}{\dimexpr\linewidth-\@tempdima-1em+\@totalleftmargin}}
\newcommand{\ensurebox}[1]{%
    \setensurelength%
    \parbox[t]{\@tempdima}{\strut #1\strut}}
\newcommand{\requirebox}[1]{%
  \setrequirelength%
  \parbox[t]{\@tempdima}{\strut #1\strut}}
\makeatother
\makeatletter
\usepackage{linegoal}
\algrenewcommand\algorithmicindent{1.5em}
\algblock[local]{}{}
\algnewcommand{\LineComment}[1]{\Statex \hskip\ALG@thistlm \parbox[t]{\@tempdima}{\hangindent=1em\hangafter=1 $\triangleright$ #1}}
\makeatother
\newlength{\continueindent}
\usepackage{etoolbox}
\makeatletter
\newcommand*{\ALG@customparshape}{\parshape 2 \leftmargin \linewidth \dimexpr\ALG@tlm+\continueindent\relax \dimexpr\linewidth+\leftmargin-\ALG@tlm-\continueindent\relax}
\apptocmd{\ALG@beginblock}{\ALG@customparshape}{}{\errmessage{failed to patch}}
\makeatother
\makeatletter
\newcommand{\algrule}[1][0.2pt]{\par\vskip.5\baselineskip\hrule height #1\par\vskip.5\baselineskip}
\makeatother

\newtheorem{theorem}{Theorem}[section]
\newtheorem{lemma}[theorem]{Lemma}

\newtheorem{corollary}[theorem]{Corollary}

\newtheorem*{definition}{Definition}

\hypersetup{colorlinks,linkcolor=blue,filecolor=blue,citecolor=blue, urlcolor=blue,pdfstartview=FitH}

\DeclareMathOperator{\poly}{poly}
\DeclareMathOperator{\ot}{\widetilde{\mathnormal{O}}}

\DeclareMathOperator*{\E}{\mathbf{E}}

\newcommand{\dsg}{DSG scheme}
\newcommand{\dgs}{DGS scheme}
\newcommand{\rsg}{RSG scheme}
\newcommand{\bor}{Bor\u{u}vka's }
\let\oldepsilon\epsilon
\renewcommand{\epsilon}{\mathnormal\varepsilon}

\title{Super-fast MST Algorithms in the Congested Clique using $o(m)$ Messages\footnote{This work is supported in part by National Science Foundation grant CCF 1318166.}}
\author{Sriram V.\ Pemmaraju \hspace{4em} Vivek B.\ Sardeshmukh\\ \small{Department of Computer Science, The University of Iowa, Iowa City, IA 52242}\\
\texttt{\{sriram-pemmaraju, vivek-sardeshmukh\}@uiowa.edu}}

\begin{document}
\maketitle
\begin{abstract}

In a sequence of recent results (PODC 2015 and PODC 2016), the running time of the fastest algorithm for the \emph{minimum spanning tree (MST)} problem in the \emph{Congested Clique} model was first improved to $O(\log \log \log n)$ from $O(\log \log n)$ (Hegeman et al., PODC 2015) and then to $O(\log^* n)$ (Ghaffari and Parter, PODC 2016). 
All of these algorithms use $\Theta(n^2)$ messages independent of the number of edges in the input graph.

This paper positively answers a question raised in Hegeman et al., and presents the first ``super-fast'' MST algorithm with $o(m)$ message complexity for input graphs with $m$ edges. 
Specifically, we present an algorithm running in $O(\log^* n)$ rounds with high probability, with message complexity $\ot(\sqrt{m \cdot n})$ 
and then build on this algorithm to derive a family of algorithms, containing for any $\varepsilon$, $0 < \varepsilon \le 1$, 
an algorithm running in $O(\log^* n/\varepsilon)$ rounds with high probability, 
using $\ot(n^{1 + \varepsilon}/\varepsilon)$ messages. 
Setting $\varepsilon = \log\log n/\log n$ leads to the first sub-logarithmic round 
Congested Clique MST algorithm that uses only $\widetilde{\mathnormal{O}}(n)$ messages. 

Our primary tools in achieving these results are 
(i) a component-wise bound on the number of candidates for MST edges, extending the sampling lemma of Karger, Klein, and Tarjan (Karger, Klein, and Tarjan, JACM 1995) and 
(ii) $\Theta(\log n)$-wise-independent linear graph sketches (Cormode and Firmani, Dist.~Par.~Databases, 2014) for generating MST candidate edges.  
\end{abstract}

\section{Introduction}
\label{section:introduction}
The \textit{Congested Clique} is a synchronous, message-passing model of distributed computing in which
the underlying network is a clique and in each round, a message of size $O(\log n)$ bits can be sent in each direction across each communication link. 
The Congested Clique is a simple, clean model for studying the obstacles imposed by congestion -- all relevant information is nearby in the network (at most 1 hop away), but may not be able to travel to
an intended node due to the $O(\log n)$-bit bandwidth restriction on the communication links.
There has been a lot of recent work in studying various fundamental problems in the Congested Clique model,
including facility location~\cite{GehweilerSPAA2006, berns2012facloc}, \textit{minimum spanning tree (MST)}
\cite{lotker2005mstJournal,Hegeman2014disc,Hegeman15podc,GhaffariParter}, shortest paths and distances \cite{Censor15podc,HolzerP14,Nanongkai14},
triangle finding \cite{drucker2012task,DolevLP12}, subgraph detection~\cite{drucker2012task},
ruling sets \cite{berns2012facloc,Hegeman2014disc}, sorting~\cite{Patt-ShamirT11,Lenzen13}, and routing \cite{Lenzen13}.
The modeling assumption in solving these problems is that the input graph $G =(V, E)$ is ``embedded''
in the Congested Clique -- each node of $G$ is uniquely mapped to a machine and the edges of $G$
are naturally mapped to the links between the corresponding machines (see Section \ref{sec:model}).

The earliest non-trivial example of a Congested Clique algorithm is the \textit{deterministic} MST
algorithm that runs in $O(\log \log n)$ rounds due to Lotker et al.~\cite{lotker2005mstJournal}.
Using \textit{linear sketching} \cite{AhnSODA12, AhnPODS12, JowhariL0, McGregorSurvey, cormode2014sampling} and 
the \textit{sampling} technique due to Karger, Klein, and 
Tarjan \cite{KKT1995MST}, Hegeman et al.~\cite{Hegeman15podc} were able to
design a substantially faster, \textit{randomized} Congested Clique MST algorithm, running in $O(\log \log \log n)$ rounds.
Soon afterwards, Ghaffari and Parter \cite{GhaffariParter} designed an $O(\log^* n)$-round algorithm, using the 
techniques in Hegeman et al., but supplemented with the use of \textit{sparsity-sensitive sketching}, which is useful
for sparse graphs and \textit{random edge sampling}, which is useful for dense graphs.
  
\paragraph{Our Contributions.} All of the MST algorithms mentioned above, essentially use the entire bandwidth of the Congested Clique model, i.e., they use $\Theta(n^2)$ messages. 
  From these examples, one might (incorrectly!) conclude that ``super-fast'' Congested Clique algorithms are only possible when the entire bandwidth of the model is used. 
  In this paper, we focus on the design of MST algorithms in the Congested Clique model that have low
  \textit{message complexity}, while still remaining ``super-fast.''
  Message complexity refers to the number of messages sent and received by all machines over 
  the course of an algorithm; in many applications, this is the dominant cost as it plays a major role 
  in determining the running time and auxiliary resources (e.g., energy) consumed by the algorithm. 
  In our main result, we present an $O(\log^* n)$-round algorithm that uses $\ot(\sqrt{m \cdot n})$
  \footnote{The notation $\ot$ hides $\poly (\log n)$ factors.}
  messages for an $n$-node, $m$-edge input graph. 
  Two points are worth noting about this message complexity
  upper bound: (i) it is bounded above by $\ot(n^{1.5})$ for all values of $m$ and is thus substantially sub-quadratic,
  independent of $m$ and (ii) it is bounded above by $o(m)$ for all values of $m$ that are 
  super-linear in $n$, i.e., when $m = \omega(n \poly (\log n))$.  
  We then extend this result to design a family of algorithms parameterized by $\varepsilon$, $0 < \varepsilon \le 1$, 
  and running in $O(\log^* n/\varepsilon)$ rounds and using $\ot(n^{1+\epsilon}/{\varepsilon})$ messages. 
  If we set $\varepsilon = \log\log n/\log n$, we get an algorithm running in $O(\log^* n \cdot \log n /\log\log n)$ rounds
  and using $\ot(n)$ messages.
  Thus we demonstrate the existence of a sub-logarithmic round MST algorithm using only $O(n \cdot \mbox{poly}(\log n))$ 
  messages, positively answering a question posed in Hegeman et al.~\cite{Hegeman15podc}.
  We note that Hegeman et al.~present an algorithm using $\ot(n)$ messages
  that runs in $O(\log^5 n)$ rounds.
  All of the round and message complexity bounds mentioned above hold with high probability (w.h.p.), i.e., with probability at least $1 - \frac{1}{n}$. 
  Our results indicate that the power of the Congested Clique model lies not so much in its 
  $\Theta(n^2)$ bandwidth as in the flexibility it provides -- any communication link that is 
  needed is present in the network, though most communication links may eventually not be needed.
  
\paragraph{Applications.} Optimizing message complexity as well as time complexity for Congested Clique algorithms has direct applications
  to the performance of distributed algorithms in other models such as the Big Data ($k$-machine) model \cite{KlauckNPR15},
  which was recently introduced to study distributed computation on large-scale graphs.
  Via a Conversion Theorem in~\cite{KlauckNPR15} one can obtain fast algorithms in the Big Data model
  from Congested Clique algorithms that have low time complexity \emph{and} message complexity. 
  Another related motivation comes from the connection between the Congested Clique model and the MapReduce model.
  In~\cite{HegemanP14} it is shown that if a Congested Clique algorithm runs in $T$ rounds and, in addition, has moderate message complexity
  then it can be simulated in the MapReduce model in $O(T)$ rounds.
  
  \subsection{Technical Preliminaries}
  \label{sec:model}
  \paragraph{Congested Clique model.}
  The \emph{Congested Clique} is a set of $n$ computing entities (nodes) connected through a complete
  network that provides point-to-point communication.
  Each node in the network has a distinct identifier of $O(\log n)$ bits. At the beginning
  of the computation, each node knows the identities of all $n$ nodes in the network and the part of the input
  assigned to it. 
  The computation proceeds in synchronous rounds. In each round each node can perform some local
  computation and send a (\textit{possibly different}) message of $O(\log n)$ bits to each of its $n-1$ neighbors.
  It is assumed that both the computing entities and the communication links are fault-free.
  The Congested Clique model is therefore specifically geared towards understanding the role of
  the limited bandwidth as a fundamental obstacle in distributed computing, in contrast to other
  classical models for distributed computing that instead focus, e.g., on the effects of latency
  (the \textsc{Local} model) or on the effects of both latency and limited bandwidth (the \textsc{Congest} model).
  
  The input graph is assumed to be a spanning subgraph of the underlying communication network. 
  Before the algorithm starts, each node knows the edges of the input graph incident on it and
  their (respective) weights.
  We assume that every edge weight can be represented with $O(\log n)$ bits.
  For ease of exposition, we assume that edge weights are distinct; 
  otherwise, without loss of generality (WLOG)
  we can ``pad'' each edge weight with the IDs of the two end points of the edge so as to distinguish the edges by weight while respecting their weight-based ordering.
  We require that when the algorithm ends, each node knows which of its incident
  edges belong to the output MST.
  
  \paragraph{Linear Sketches.}
  A key tool used by our algorithm is \textit{linear sketches}~\cite{AhnSODA12, AhnPODS12, McGregorSurvey}.
  Let $\mathbf{a}_v$ denote a vector whose non-zero entries represent edges incident on $v$.
  A \textit{linear sketch} of $\mathbf{a}_v$ is a low-dimensional random vector $\mathbf{s}_v$,
  typically of size $O(\poly (\log n))$, with two properties:
  (i) sampling from the sketch $\mathbf{s}_v$ returns a non-zero entry of $\mathbf{a}_v$ with uniform probability (over all non-zero entries in $\mathbf{a}_v$) and
  (ii) when nodes in a connected component are merged, the sketch of the new ``super node'' is obtained by coordination-wise addition of the sketches of
  the nodes in the component.
  The first property is referred to as \textit{$\ell_0$-sampling} in the streaming 
  literature~\cite{cormode2014sampling,McGregorSurvey, JowhariL0} and
  the second property is \textit{linearity}.
  The graph sketches used in~\cite{AhnSODA12, AhnPODS12, McGregorSurvey} rely on the $\ell_0$-sampling algorithm by Jowhari et al.~\cite{JowhariL0}.
  Sketches constructed using the Jowhari et al.~\cite{JowhariL0} approach use $\Theta(\log^2 n)$ bits per sketch,
  but require polynomially many mutually independent random bits to be shared among all nodes in the network.
  Sharing this volume of information is not feasible; it takes too many rounds and too many messages.
  So instead, we appeal to the $\ell_0$-sampling algorithm of Cormode and Firmani~\cite{cormode2014sampling} 
  which requires a family of
  $\Theta(\log n)$-wise independent hash functions, as opposed to hash functions with full-independence.
  Hegeman et al.~\cite{Hegeman15podc} provide details of how the Cormode-Firmani approach can be used in the Congested Clique
  model to construct graph sketches.
  We summarize their result in the following theorem.
  \begin{theorem}[Hegeman et al.~\cite{Hegeman15podc}]\label{thm:sketches}
  Given an input graph $G=(V, E)$, $n = |V|$, there is a Congested Clique algorithm running in $O(1)$ rounds and
  using $O(n \cdot \poly (\log n))$ messages,
  at the end of which every node $v \in V$ has computed a linear sketch $\mathbf{s}_v$ of $\mathbf{a}_v$.
  The size of the computed sketch of a node is $O(\log^4 n)$ bits.
  The $\ell_0$-sampling algorithm on sketch $\mathbf{s}_v$ succeeds with probability at least $1-n^{-2}$ and, 
  conditioned on success, returns an edge in $\mathbf{a}_v$ with probability in the range $[1/L_v - n^{-2}, 1/L_v + n^{-2}]$,
  where $L_v$ is the number of non-zero entries in $\mathbf{a}_v$.
  \end{theorem}
  \paragraph{Concentration Bounds for sums of $k$-wise-independent random variables.}
  The use of $k$-wise-independent random variables, for $k = \Theta(\log n)$, plays a key role in keeping the
  time and message complexity of our algorithms low.
  The use of $\Theta(\log n)$-wise independent hash functions in the construction of linear sketches has been
  mentioned above.
  In the next subsection, we discuss the use of $\Theta(\log n)$-wise-independent edge sampling
  as a substitute for the fully-independent edge sampling of Karger, Klein, and Tarjan.
  For our analysis we use the following concentration bound on the sum of $k$-wise independent random
  variables, due to Schmidt et al.~\cite{SchmidtSS95} and slightly simplified by Pettie and Ramachandran~\cite{PettieRamachandran}. 
  \begin{theorem}[Schmidt et al.~\cite{SchmidtSS95}]\label{thm:schmidt}
  Let $X_1, X_2, \ldots, X_n$ be a sequence of random
  $k$-wise independent 0-1 random variables with $X = \sum_{i=1}^n X_i$. If $k \ge 2$ is even and 
  $C \ge \E[X]$ then:
  $$Pr(|X - \E[X]| \ge T)  \le \left[\sqrt{2} \cosh\left(\sqrt{k^3/36C}\right)\right] \cdot \left(\frac{kC}{eT^2}\right)^{k/2}.$$
  \end{theorem}
  We use the above theorem for $k = \Theta(\log n)$ and $C = T = \E[X]$. Furthermore, in all instances in which we use this bound,
  $\E[X] > k^3$ and therefore the contribution of the $\cosh(\cdot)$ term is $O(1)$, whereas the contribution of the second term on the right hand side
  is smaller than $1/n^c$ for any constant $c$.
  \paragraph{MST with Linear Message Complexity.}
  The ``super-fast'' MST algorithms mentioned so far \cite{lotker2005mstJournal,Hegeman15podc,GhaffariParter} use $\Theta(n^2)$ messages, independent of the number of edges in
  the input graph. One reason for this is that these algorithms rely on deterministic constant-round Congested Clique 
  algorithms for routing and sorting due to Lenzen \cite{Lenzen13}. 
  Lenzen's algorithms do not attempt to explicitly conserve messages and need 
  $\Omega(n^{1.5})$ messages independent of the number of messages being routed or the number of keys being sorted.
  However, the above-mentioned MST algorithms do not need the full power of Lenzen's algorithms. 
  We design sorting and routing protocols that work in slightly restricted settings, but use only a linear number
  of messages (i.e., linear in the total number messages to be routed or keys to be sorted). Details of these protocols
  appear in Section~\ref{sec:linear}. 
  We use these protocols (instead of Lenzen's protocols) as subroutines in the Ghaffari-Parter MST algorithm 
  \cite{GhaffariParter} to derive a version that uses only linear (up to a polylogarithmic factor) number of messages. 
  
  \subsection{Algorithmic Overview} \label{sub:overview}
  The high-level structure of our algorithm is simple.
  Suppose that the input is an $n$-node, $m$-edge graph $G = (V, E)$.
  We start by sparsifying $G$ by sampling each edge with probability $p$ and compute a \textit{maximal minimum weight spanning forest}
  $F$ of the resulting sparse subgraph $H$.
  Thus $H$ contains $O(m \cdot p)$ edges w.h.p.
  Now consider an edge $\{u, v\}$ in $G$ and add it to $F$; if $F + \{u, v\}$ contains a cycle and $\{u, v\}$ is
  the heaviest edge in this cycle, then by Tarjan's ``red rule'' \cite{TarjanBook} the MST of $G$ does not contain edge $\{u, v\}$.
  Ignoring all such edges leaves a set of edges that are candidates for being in the MST.
  We appeal to the well-known sampling lemma due to Karger, Klein, and Tarjan \cite{KKT1995MST} (KKT sampling) that provides an 
  estimate of the size of this set of candidates.
  \begin{definition}[$F$-light edge~\cite{KKT1995MST}]
  Let $F$ be a forest in a graph $G$ and let $F(u,v)$ denote the path (if any) connecting $u$ and $v$ in $F$.
  Let $w_F(u, v)$ denote the maximum weight of an edge on $F(u, v)$ (if there is no path then $w_F(u, v) = \infty$).
  We call an edge $\{u, v\}$ \emph{$F$-heavy} if $w(u, v) > w_F(u, v)$, and \emph{$F$-light} otherwise.
  \end{definition}
  \begin{lemma}[KKT Sampling Lemma~\cite{KKT1995MST}]
  \label{lemma:kkt}
  Let $H$ be a subgraph obtained from $G$ by including each edge independently
  \footnote{For reasons that will become clear later, our goal of keeping the
    message complexity low, does not allow us to assume full independence in this sampling. 
    Instead we use $\Theta(\log n)$-wise independent sampling and show that a slightly weaker version of the KKT Sampling Lemma holds even with limited independence sampling.} 
  with probability $p$ and let $F$ be the maximal minimum weight spanning forest of $H$.
  The number of $F$-light edges in $G$ is at most $n/p$, w.h.p.
\end{lemma}
As our next step we compute the set of $F$-light edges and in our final step, we compute an MST of the subgraph
induced by the $F$-light edges.
Thus, at a high level, our algorithm consists of two calls to an MST subroutine on sparse graphs,
one with $O(m \cdot p)$ edges and the other with $O(n/p)$ edges.
In between, these two calls is the computation of $F$-light edges.
This overall algorithmic structure is clearly visible in Lines 5--7 in the pseudocode in Algorithm~\ref{algo:mst1} \textsc{MST-v1}.

There are several obstacles to realizing this high-level idea in the Congested Clique model in order to
obtain an algorithm that is ``super-fast'' and yet has low message complexity.
The reason for sparsifying $G$ and appealing to the KKT Sampling Lemma is the expectation that we would need
to use fewer messages to compute an MST on a sparser input graph.
However, all of the ``super-fast'' MST algorithms mentioned earlier in the paper use $\Theta(n^2)$ messages
and are insensitive to the number of edges in the input graph.
In our first contribution, we develop a collection of simple, low-message-complexity distributed routing and
sorting subroutines that we can use in any of the ``super-fast'' MST algorithms mentioned above \cite{lotker2005mstJournal,Hegeman15podc,GhaffariParter} (see Section~\ref{sec:linear}) in order to reduce their message 
complexity to $O(m)$, without increasing their time complexity.
Specifically, modifying the Ghaffari-Parter MST algorithm to use these routing and sorting subroutines 
allows us to complete the two calls to the MST subroutine in $O(\log^* n)$ rounds using 
$\max\{O(m \cdot p), O(n/p)\}$ messages.
Setting the sampling probability $p$ in our algorithm 
to $\sqrt{\frac{n}{m}}$ balances the two terms in the $\max(\cdot, \cdot)$ and yields a
message complexity of $O(\sqrt{m \cdot n})$.
We describe this in Section~\ref{sec:linear}.

Our second and \emph{main} contribution (Section \ref{sec:computeF})
is to show that the computation of $F$-light can be completed
in $O(1)$ rounds, while still using $\ot(\sqrt{m \cdot n})$ messages.
To explain the challenge of this computation we present two simple algorithmic scenarios:
\begin{itemize}
  \item Suppose that we want each node $u$ to perform a local computation to determine which of its incident edges
  from $G$ are $F$-light.
  To do this, node $u$ needs to know $w_F(u, v)$ for all neighbors $v$. Thus $u$ needs $\mbox{degree}_G(u)$
  pieces of information and overall this approach seems to require the movement of $\Omega(m)$ pieces of
  information, i.e., $\Omega(m)$ messages.
  \item Alternately, we might want each node that knows $F$ to be responsible for determining which
  edges in $G$ are $F$-light.
  In this case, the obvious approach is to send queries of the type ``Is edge $\{u, v\}$ $F$-light?'' to nodes that know $F$.
  This approach also requires $\Omega(m)$ messages.
\end{itemize}
Various combinations of and more sophisticated versions of these ideas also require $\Omega(m)$ 
messages.
So the fundamental question is how do we determine the status (i.e., $F$-light or $F$-heavy) of $m$ edges while 
exchanging far fewer than $m$ messages?
Below we outline two techniques we have developed in order to answer this question.
\begin{description}
  \item [Component-wise bound on number of $F$-light edges.] 
  As mentioned above, the KKT Sampling Lemma upper bounds the total number of $F$-light edges
  by $O(n/p)$, which is $O(\sqrt{m \cdot n})$ for $p = \sqrt{n/m}$.
  We show (in Corollary~\ref{coro:fsize}) that a slightly weaker bound (weaker by a logarithmic factor) holds even if the edge-sampling is done using an $\Theta(\log n)$-wise-independent sampler.
  If we could ensure that the total volume of communication is proportional to the number of $F$-light edges, 
  we would achieve our goal of $o(m)$ message complexity.
  To achieve this goal we show that the set of $F$-light edges has additional structure; they are ``evenly distributed'' over the components of $F$.
  To understand this imagine that $F$ is constructed from $H$ using Bor\u{u}vka's algorithm.
  Let $\mathcal{C}^i = \{C^i_1, C^i_2, \ldots\}$ be the set of components at the beginning of a phase $i$ of the algorithm.
  For each component $C^i_j \in \mathcal{C}^i$, the algorithm picks a \textit{minimum weight outgoing edge (MWOE)} $e^i_j$ from $F$. 
  Components are merged using edges $e^i_j, j = 1, 2,\ldots$ and we get a new set of components $\mathcal{C}^{i+1}$.
  Let $L^i_j$ be the set of edges in $G$ leaving component $C^i_j$ \textit{with weight at most $w(e^i_j)$}.
  We show in Lemma \ref{lemma:correctness} that the set of all $F$-light edges is just the union
  of the $L^i_j$'s, over all phases $i$ and components $j$ within Phase $i$.
  Furthermore, we show in Lemma \ref{lemma:lsize} that the size of $L^i_j$ for any $i, j$ is
  is bounded by $\ot(1/p)$ w.h.p.
  This ``even distribution'' of $F$-light edges suggests that we could make each component $C^i_j$ responsible for identifying the 
  $L^i_j$-edges.
  Note that we don't use distributed \bor algorithm to compute $F$ because that would take $\Theta(\log n)$ rounds.
  We compute $F$ in $O(\log^* n)$ rounds using \textsc{LinearMessages-MST}, the modified Ghaffari-Parter algorithm (see Section~\ref{sec:linear}).
  $F$ is then gathered at each of a small number of nodes and each node who knows $F$ completely simulates \bor algorithm \textit{locally}
  on $F$, thus identifying the components $C^i_j$ and their MWOE's $e^i_j$.)
  
  \item[Component-wise generation of $F$-light edges using linear sketches.] 
  Linear sketches play a key role in helping nodes in each component $C^i_j$ collectively 
  compute all edges in $L^i_j$. 
  For any node $v$ and number $x$, let $N_x(v)$ denote the set of neighbors of $v$ that are 
  connected to $v$ via edges of weight less than $x$.
  Each node $v \in C^i_j$ computes a \textit{$w(e^i_j)$-restricted sketch} $\mathbf{s}_v$,
  i.e., a sketch of its neighborhood $N_{w(e^i_j)}$,
  and sends it to the component leader of $C^i_j$ who aggregates these sketches to compute a single component sketch. 
  Sampling this sketch yields a single edge in $L^i_j$. 
  Since $L^i_j$ has $\ot(1/p)$ edges, each node $v\in C^i_j$ can send $\ot(1/p)$ separate 
  $w(e^i_j)$-restricted sketches to the component leader of $C^i_j$ and the 
  Coupon Collector argument ensures that this volume of sketches is enough to generate 
  \textit{all} edges incident in $L^i_j$ w.h.p. 
\end{description}

\paragraph{Remark:} 
The sampling approach of Karger, Klein, and Tarjan is used in a somewhat minor way in earlier 
Congested Clique MST algorithms \cite{GhaffariParter,Hegeman15podc} and in fact in \cite{KorhonenArxiv16} it is shown that this sampling
approach can be replaced by a simple, deterministic sparsification.
However, KKT sampling and specifically its $\Theta(\log n)$-wise independent version that we use in the current algorithm
seems crucial for ensuring low message complexity, while keeping the algorithms fast.

\subsection{Related Work}

It is important to point out that our algorithms are designed for the so-called KT1 \cite{peleg2000distributed} model,
where every node initially knows the IDs of all its neighbors, in addition to its own ID.
(In the Congested Clique model, this means that each node knows the IDs of all $n$ nodes in the network.)
If we drop this assumption and work in the so-called KT0 model \cite{peleg2000distributed}, in which nodes are unaware of IDs 
of neighbors, then it has been shown in \cite{Hegeman15podc} that $\Omega(m)$ messages are needed by any Congested Clique MST algorithm (including randomized Monte Carlo algorithms, and regardless of
the number of rounds) on an $m$-edge input graph.
In fact, this lower bound is shown for the simpler graph connectivity problem.

There have also been some recent developments on simultaneously optimizing message complexity and round complexity for the
MST problem in the \textsc{Congest} model.
For example, in \cite{PanduranganRS16} it is shown that there exists a randomized (Las Vegas) algorithm that runs in 
$\ot(\sqrt{n} + \mbox{diameter}(G))$ rounds and uses $\ot(m)$ messages (both w.h.p.).
This improves the message complexity of the well-known Kutten-Peleg algorithm \cite{KuttenPeleg1998}, without sacrificing round complexity (upto
polylogarithmic factors).
The Kutten-Peleg algorithm runs in $O(\sqrt{n}\log^* n + \mbox{diameter}(G))$ rounds, while using 
$O(m + n^{1.5})$ messages.
Note that the algorithm in \cite{PanduranganRS16} simultaneously matches the round complexity lower bound 
\cite{Elkin2006,DasSarmaSICOMP2011} and the message complexity lower bound \cite{KuttenPPRTJACM2015} for
the MST problem.

The above-mentioned upper and lower bound results assume the KT0 model.
In the KT1 model, the message complexity lower bound of Kutten et al.~\cite{KuttenPPRTJACM2015} does not hold and 
King et al.~\cite{king15testout} were able to design an MST algorithm in the KT1 \textsc{Congest} model that uses 
$\ot(n)$ messages, though this algorithm has significantly higher round complexity than
$\ot(\sqrt{n} + \mbox{diameter}(G))$ rounds.

As mentioned earlier, Hegeman et al.~\cite{Hegeman15podc} present a Congested Clique MST algorithm using 
$\ot(n)$ messages, but running in $O(\log^5 n)$ rounds.
One can make a few changes to the King et al.~\cite{king15testout} \textsc{Congest}-model algorithm to implement it in the Congested Clique model, 
requiring $\ot(n)$ messages, but running in $O(\log^2 n/\log \log n)$ rounds. 

\section{MST Algorithms}\label{sec:mst}
In this section we describe two ``super-fast'' MST algorithms,
the first runs in $O(\log^* n)$ rounds, using $\ot(\sqrt{m\cdot n})$ messages and
the second algorithm running in $O(\log^* n/\varepsilon)$ rounds, 
using $\ot(n^{1+\varepsilon}/\varepsilon)$ messages, for any $0 < \varepsilon \le 1$. 

\subsection{A super-fast algorithm using \texorpdfstring{$\ot(\sqrt{mn})$}{o(m)} messages}
\label{sub:mst1}
Our first algorithm \textsc{MST-v1}, shown in Algorithm~\ref{algo:mst1} has already been outlined in Section~\ref{sub:overview}. 
The correctness, time complexity, and message complexity of this algorithm depends mainly on two subroutines: $\textsc{LinearMessages-MST}(\cdot)$ and $\textsc{Compute-F-Light}(\cdot)$.
For the purpose of this section, we assume that $\textsc{LinearMessages-MST}(H)$ computes an MST on an $n$-node $m$-edge input graph $H$ in $O(\log^* n)$ rounds using $\ot(m)$ messages. 
This is shown in Section~\ref{sec:linear}. 
We also show that $\textsc{Compute-F-Light}(G, F, p)$ terminates in $O(1)$ rounds using $\ot(n/p)$ messages w.h.p. 
This is the main result in our paper and is shown in Section~\ref{sec:computeF}.
\begin{algorithm}[!ht]
  \caption{\textsc{MST-v1} \label{algo:mst1}}
  \begin{algorithmic}[1]
    \Require \requirebox{An edge-weighted $n$-node, $m$-edge graph $G = (V, E, w)$.
	\LineComment{Each node knows weights and end-points of incident edges.
	Every weight can be represented using $O(\log n)$ bits.}}
    \Ensure \ensurebox{An MST $\mathcal{T}$ of $G$.
	\LineComment{Each node in $V$ knows which of its incident edges are part of $T$.}}
	\textcolor{gray!70}{\algrule}
    \LineComment{Let $v^*$ denote the node with lowest ID in $V$, known to all nodes.}
    \State $v^*$ generates a sequence $\pi$ of $\Theta(\log^2 n)$ bits independently and uniformly at random and shares with all nodes in $V$. 
    \State $p \gets \sqrt{\frac{n}{m}}$
    \State Each node constructs an $\Theta(\log n)$-wise-independent sampler from $\pi$ and uses this to sample each incident edge in $G$ with probability $p$
    \State $H \gets$ the spanning subgraph of $G$ induced by the sampled edges
    \State $F \gets \textsc{LinearMessages-MST}(H)$
    \State $E_\ell \gets \textsc{Compute-F-Light}(G, F, p)$
    \State $\mathcal{T} \gets \textsc{LinearMessages-MST}((V, E_\ell, w))$
    \State \Return $\mathcal{T}$
  \end{algorithmic}
\end{algorithm}
\begin{lemma} \label{lemma:hsize}
  For some constants $c_1, c_2 > 1$, 
  (i) $\Pr(|E(H)| > c_1 \cdot \sqrt{mn}) < \frac{1}{n}$ and
  (ii) $\Pr(|E_\ell| > c_2 \cdot \sqrt{mn} \poly (\log n)) < \frac{1}{n}$. 
\end{lemma}
\begin{proof}
  For $0 < i \leq m$,  let $X_i = 1$ if edge $i$ is sampled. 
  Hence $|E(H)| = \sum_i X_i$ and $\E[|E(H)|] = \sqrt{mn}$.
  Note that $X_i$'s are $\Theta(\log n)$-wise independent. 
  Therefore, by Theorem~\ref{thm:schmidt} we have, $\Pr(|E(H)| > c_1\sqrt{mn}) < \frac{1}{n}$ for some suitable constant $c_1 > 1$. 
  Claim (ii) follows from Corollary~\ref{coro:fsize}.
\end{proof}
The following theorem summarizes the properties of Algorithm \textsc{MST-v1}. 
The running time and message complexity bounds follow from Table~\ref{tab:mst1}.
\begin{table}[b]
  \caption{Time and message complexity for steps in Algorithm~\ref{algo:mst1} \textsc{MST-v1}\label{tab:mst1}}
  \centering
  \small
  \begin{tabular}{l l l l}
    \hline
    \textbf{Step} & \textbf{Time} & \textbf{Messages} & \textbf{Analysis} \\ \hline \noalign{\vskip 0.2ex} 
    1 & $O(1)$ & $\ot(n)$ & Theorem~\ref{thm:dgs} \\  
    2-4 & - & - & Local computation \\
    5 & $O(\log^* n)$ & $\ot(|E(H)|)$ & Theorem~\ref{thm:linearMST} \\
    6 & $O(1)$ & $\ot\left(\sqrt{mn}\right)$ & Theorem~\ref{thm:computeF} with $p=\sqrt{\frac{n}{m}}$ \\
    7 & $O(\log^* n)$ & $\ot(|E_\ell|)$ & Theorem~\ref{thm:linearMST} \\ \hline
  \end{tabular}
\end{table}
\begin{theorem}\label{thm:mst1}
  Algorithm \textsc{MST-v1} computes an MST of an edge-weighted $n$-node, $m$-edge graph $G$ when it terminates.
  Moreover, it terminates in $O(\log^* n)$ rounds and requires $\ot(\sqrt{mn})$ messages w.h.p.
\end{theorem}

\subsection{Trading messages and time}
\label{sub:mst2}
The \textsc{MST-v2} algorithm (shown in Algorithm~\ref{algo:mst2}) is a recursive version of \textsc{MST-v1} algorithm 
yielding a time-message trade-off.
The algorithm recurses until the number of edges in the subproblem becomes ``low'' enough to solve it via a call to the \textsc{LinearMessages-MST} subroutine. 
Specifically, we treat a $n$-node graph with $m = O(n^{1+\epsilon})$ edges as a base case.
For graphs with more edges we use a sampling probability of $p = 1/n^\epsilon$, leading to a sparse graph $H$ with $O(m/n^{\epsilon})$ edges w.h.p., which is recursively processed. 
The use of limited independence sampling is critical here.
One simple approach to sampling an edge would be to let the endpoint with higher ID sample the edge and inform the other endpoint \textit{if the outcome is positive}. 
Unfortunately, this would lead to the use of $\ot(m/n^{\epsilon})$ messages w.h.p., exceeding our target of $\ot(n^{1+\epsilon})$ messages when $m$ is large\footnote{This approach would have worked fine for \textsc{MST-v1}, but to keep the two algorithms consistent to the extent possible, we use the $\Theta(\log n)$-wise
independent sampler there as well.}.
Using $\Theta(\log n)$-wise-independent sampling allows us to complete the sampling step using $\ot(n)$ messages. 
\begin{algorithm}[!ht]
  \caption{\textsc{MST-v2} \label{algo:mst2}}
  \begin{algorithmic}[1]
    \Require \requirebox{An edge-weighted $n$-node, $m$-edge graph $G = (V, E, w)$
	\LineComment{Each node knows weights and end-points of incident edges in $G$.
	    Every weight can be represented using $O(\log n)$ bits. 
	    There is a parameter $0 < \epsilon \le 1$, known to all nodes.}}
    \Ensure \ensurebox{An MST $\mathcal{T}$ of $G$. 
	\LineComment{Each node in $V$ knows which of its incident edges are part of $T$.}}
	\textcolor{gray!70}{\algrule}
    \LineComment{Let $v^*$ denote the node with lowest ID in $V$ and $c \geq 1$ is a constant.}
    \If{$m < c \cdot n^{1+\epsilon}$}
	\State $\mathcal{T} \gets \textsc{LinearMessages-MST}(G)$
	\State \Return $\mathcal{T}$
    \Else
	\State $v^*$ generates a sequence $\pi$ of $\Theta(\log^2 n)$ bits independently and uniformly at random and shares with all nodes in $V$ 
	\State $p \gets 1/n^\epsilon$
	\State Each node constructs an $\Theta(\log n)$-wise-independent sampler from $\pi$ and uses this to sample each incident edge in $G$ with probability $p$
	\State $H \gets$ the spanning subgraph of $G$ induced by the sampled edges
	\State $F \gets \textsc{MST-v2}(H)$
	\State $E_\ell \gets \textsc{Compute-F-Light}(G, F, p)$ 
	\State $\mathcal{T} \gets \textsc{LinearMessages-MST}((V, E_\ell, w))$
	\State \Return $\mathcal{T}$ 
    \EndIf
  \end{algorithmic}
\end{algorithm}

\begin{theorem}
  Algorithm \textsc{MST-v2} outputs an MST of an edge-weighted $n$-node, $m$-edge graph when terminates. 
  Moreover, for any $\epsilon > 0$, it terminates after $O\left({\log^* n}/{\epsilon}\right)$ rounds and uses $\ot\left({n^{1+\epsilon}}/{\epsilon}\right)$ messages, w.h.p. 
\end{theorem}
\begin{proof}
  If $m = O(n^{1+\epsilon})$ then the claim follows from Theorem~\ref{thm:linearMST}. 
  Let $T(m)$ denote the time required for Algorithm~\ref{algo:mst2} to compute an MST of a $n$-node, $m$-edge graph.
  Since $\textsc{Compute-F-Light}(\cdot)$ runs in $O(1)$ time and $\textsc{LinearMessages-MST}(\cdot)$ runs in $O(\log^* n)$ time, we see that,  
  $T(m) = T(m/n^{\epsilon}) + O(\log^* n)$, for all large $m$. 
  The first quantity is the result of a recursive call on the sampled graph $H$, where each edge is sampled with probability $p = 1/n^{\epsilon}$. 
  Solving this recursion with base case $m = O(n^{1+\epsilon})$, we get $T(m) = O(\log^* n / \epsilon)$. 
  The message complexity bound is obtained by similar arguments. 
\end{proof}
Setting $\epsilon = \log\log n / \log n$, we get the following result. 
\begin{corollary}
  There exists an algorithm that computes an MST of an $n$-node, $m$-edge input graph and w.h.p. terminates in $O(\log n \cdot \log^* n / \log \log n)$ rounds and $\ot(n)$ messages.
\end{corollary}

\section{Efficient Computation of \texorpdfstring{$F$}{F}-light Edges}
\label{sec:computeF}
In this section we describe the \textsc{Compute-F-Light} algorithm and prove its correctness and analyze its time and message complexity. 
The inputs to this algorithm are the graph $G$, a spanning forest $F$ of $G$, and a probability $p$. 
Recall that $F$ is the maximal minimum weight spanning forest of the subgraph $H$ obtained by sampling edges in $G$ with probability $p$, using a $\Theta(\log n)$-wise-independent sampler.
The main ideas in \textsc{Compute-F-Light} have been informally described in Section~\ref{sub:overview}.
The \textsc{Compute-F-Light} algorithm is described below in Algorithm~\ref{algo:f-light}.
\begin{algorithm}[!ht]
  \caption{\textsc{Compute-F-Light} \label{algo:f-light}}
  \begin{algorithmic}[1]
    \Require \requirebox{(i) An edge-weighted $n$-node, $m$-edge graph $G = (V, E, w)$, 
	(ii) A spanning forest $F$ of $G$, and 
	(iii) a number $p$, $0 < p < 1$. 
	\LineComment{$F$ is a maximal minimum weight spanning forest of a subgraph $H$ of $G$, 
	    where $H$ is a spanning subgraph of $G$ obtained by sampling each edge in $G$ with probability $p$ using a $\Theta(\log n)$-wise-independent sampler. 
	    Each node knows weights and end-points of incident edges from $G$ and $F$. 
	    Every weight can be represented using $O(\log n)$ bits.}}
    \Ensure \ensurebox{$F$-light edges of $G$. 
	\LineComment{Each node in $V$ knows which of its incident edges from $G$ are $F$-light.}}
	\textcolor{gray!70}{\algrule}
    \State Let $\{v_1, v_2, \ldots, v_c\}$ be set of \emph{commander nodes} (or in short, \emph{commanders}) where $c = \Theta(\log n)$. 
	Gather $F$ at each of these commanders. 
    \State Each commander simulates Bor\u{u}vka's algorithm locally on input graph $F$. 
	Let $\mathcal{C}^i = \{C^i_1, C^i_2, \ldots\}$ be the set of components at the beginning of Phase $i$. 
	The node with smallest ID in a component $C^i_j$ is the \textit{leader} of component $C^i_j$ and
	the ID of the leader serves as the label of each component.
	For each component $C^i_j \in \mathcal{C}^i$, the algorithm picks a MWOE $e^i_j$ from $F$.  
	Components are merged and we get a new set of components $\mathcal{C}^{i+1}$.
        If there is no incident edge on a component $C^i_j$ in $F$ then commander sets $e^i_j = \bot$ with the understanding that $w(\bot) = \infty$. 
    \State For each component $C^i_j$, commander $v_i$ sends the following 3-tuple to each node in $C^i_j$:
      \Statex{\hspace{\algorithmicindent}} (a) Phase number $i$,
      	(b) label of $C^i_j$, and  
	(c) $w(e^i_j)$.
    \State A node $v$ having received a 3-tuple $(i, \ell, w')$ associated with component $C^i_j$ for some $i$ and $j$ 
	computes $\Theta\left(\frac{\log^5 n}{p}\right)$ different graph sketches with respect to its $w'$-restricted neighborhood $N_{w'}(v)$. 
    \State The component leader of $C^i_j$ for each $i$ and $j$, gathers $\Theta\left(\frac{\log^5 n}{p}\right)$ $w(e^i_j)$-restricted sketches from all the nodes in $C^i_j$ and computes $w(e^i_j)$-restricted sketches of $C^i_j$. 
    	Then it samples an edge from each sketch computed and notifies the end-points of all sampled edges. 
    \State \Return Union of sampled edges over all $i$ over all $j$. 
  \end{algorithmic}
\end{algorithm}
\subsection{Analysis} 
Let $\mathcal{C}^i = \{C^i_1, C^i_2, \ldots\}$ be the set of components at the beginning of Phase $i$ of \bor algorithm being locally  simulated on $F$. 
Consider the set of edges from $G$ with exactly one endpoint in $C^i_j$ with weight at most $w(e^i_j)$: 
$L^i_j = \{e=\{u, v\} \in E \mid u \in C^i_j, v \notin C^i_j \mbox { and } w(e) \leq w(e^i_j) \}$. 
For example, see Figure~\ref{fig:computeF}.
\begin{figure}[t]
  \centering
  \includegraphics[width=\textwidth]{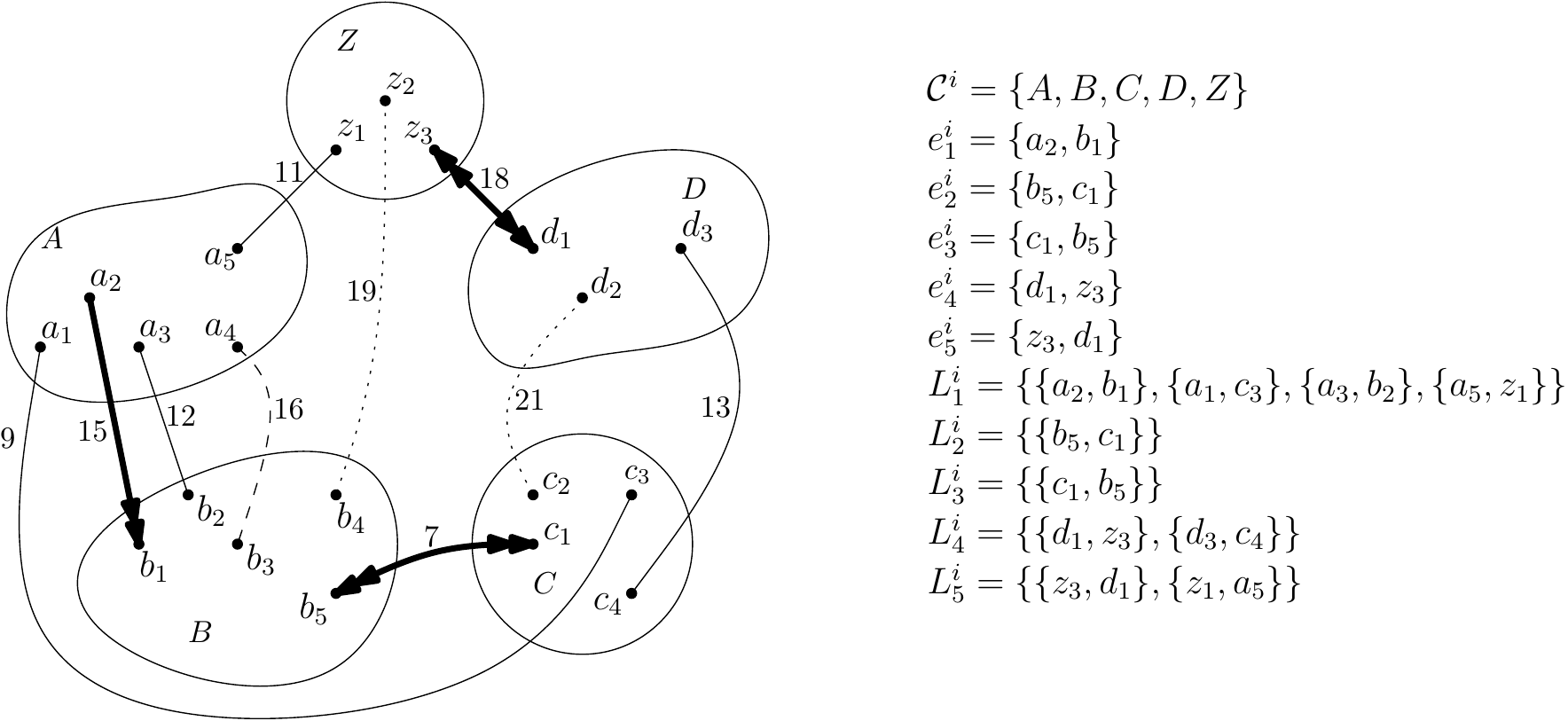}
  \caption{Illustration of notation and terminology used in Algorithm~\ref{algo:f-light}     \textsc{Compute-F-Light}. 
    At the beginning of Phase $i$ of \bor algorithm, there are 5 components $\{A, B, C, D, Z\}$. 
    Each component's MWOE in $F$ is shown as thick directed arc. 
    Solid arcs show edges in $G$ that are in respective $L^i_j$'s and hence identified as being $F$-light. 
    Dashed arcs (e.g., $a_4b_3$) represent edges that the algorithm ignores; these edge are not $F$-light.
    Dotted arcs (e.g., $b_4z_2, c_2d_2$) represent edges in $G$ whose status has not yet been resolved by the algorithm. 
    After the merging of components is completed, we end up with two components $\{ABC, DZ\}$. 
    \label{fig:computeF}}
\end{figure}
Our first task is to bound the size of $L^i_j$ and for this we appeal to the following lemma from Pettie and  Ramachandran~\cite{PettieRamachandran} on sampling from an ordered set. 
\begin{lemma}[Pettie \& Ramachandran~\cite{PettieRamachandran}] \label{lemma:PR}
Let $\chi$ be a set of $n$ totally ordered elements and $\chi_p$ be a subset of $\chi$, 
derived by sampling each element with probability $p$ using a $k$-wise-independent sampler. 
Let $Z$ be the number of unsampled elements less than the smallest element in $\chi_p$. Then
$\E[Z] \leq  p^{-1}(8(\pi/e)^2 + 1)$ for $k \geq 4$.
\end{lemma}
Observe that a straight-forward application of the above lemma gives us $\E[|L^i_j|] = O(1/p)$. 
In the next lemma, we modify the proof of Lemma~\ref{lemma:PR} in Pettie \& Ramachandran~\cite{PettieRamachandran} to obtain a bound on size of $L^i_j$ that holds w.h.p. 
\begin{lemma} \label{lemma:lsize}
  $\Pr\left(\mbox{There exist $i$ and $j$:} |L^i_j| > c\cdot \log^3 n / p\right) < \frac{1}{n}$ for some constant $c > 1$.
\end{lemma}
\begin{proof}
  Fix a Phase $i$ and a component $C^i_j$ in that phase.
  Let $X$ be the set of all edges from $G$ having exactly one endpoint in $C^i_j$.
  Let $X_t$ be an indicator random variable defined as $X_t = 1$ if the $t^{th}$ smallest edge in $X$ is sampled, and 0 otherwise.
  For any integer $\ell$, $1\leq \ell \leq |X|$, let $S_\ell = \sum_{t=1}^\ell X_t$ count the number of ones in $X_1, \ldots, X_\ell$.
  Note that $L^i_j \subseteq X$ is a set of all edges with weight at most $e^i_j$, the MWOE from $C^i_j$ in $F$. 
  This implies that the lightest edge in $X$ that is sampled is $e^i_j$, otherwise \bor algorithm would have chosen a different MWOE. 
  In other words, $X_k = 0$ for all $k \leq \ell$ if the rank of $e^i_j$ in the ordered set $X$ is $\ell + 1$ or more. 
  Therefore, $\Pr\left(|L^i_j| > \ell\right) = Pr(S_\ell = 0)$. 

  Observe that, $S_\ell$ is a sum of 0-1 random variables which are $\Theta(\log n)$-wise-independent and $\E[S_\ell] = p\ell$.
  By Theorem~\ref{thm:schmidt}, we have $\Pr(S_\ell = 0) < \frac{1}{n^3}$ for $\ell > c \cdot \log^3 n / p$ for some constant $c > 1$. 
  The lemma follows by applying union bound over all phases and components. 
\end{proof}
\begin{lemma} \label{lemma:lsample}
  For any Phase $i$ and any component-MWOE pair $(C^i_j, e^i_j)$, 
  w.h.p. $O\left(\log^5 n / p\right)$ $w(e^i_j)$-restricted sketches of $C^i_j$ are sufficient to find all edges in $L^i_j$.
\end{lemma}
\begin{proof}
  Consider an oracle which when queried returns an edge in $L^i_j$ independently and uniformly at random. 
  Let $T_s$ denote the number of the oracle queries required to obtain $s = |L^i_j|$ distinct edges (i.e., all edges in $L^i_j$).
  Then by the Coupon Collector argument~\cite{MotwaniRaghavan}, $Pr(T_s > \beta s \log s) < s^{-\beta + 1}$ for any $\beta > 1$.
  Also, if the oracle is not uniform, but is ``almost uniform,'' returning an edge in $L^i_j$ with probability $\frac{1}{s} \pm s^{-\alpha}$ for a constant $\alpha > 2$, then we get $Pr(T_s > \beta s \log s + o(1)) < s^{-\beta + 1}$. 

  Now, to simulate a $t^{th}$ oracle query ($t \in [1, T_s]$) mentioned above, we sample an unused sketch of $C^i_j$ until we get an edge.
  Since sampling from a sketch fails with probability at most $n^{-2}$, w.h.p., $O(1)$ sketches are sufficient to simulate one oracle query. 
  Hence w.h.p., $O(T_s)$ sketches are sufficient to simulate $T_s$ oracle queries.
  Therefore, with probability at least $1 - s^{-\beta + 1}$, $O(\beta s \log s)$ sketches are sufficient to get $s$ distinct edges from $L^i_j$. 

  By Lemma~\ref{lemma:lsize}, we have w.h.p., $s = |L^i_j| = O\left(\log^3 n/p\right)$. 
  Therefore by letting $s = \Theta\left(\log^3 n/p\right)$ and $\beta = O(\log n)$ in the above argument, w.h.p., $O\left(\log^5 n/p\right)$ sketches are sufficient to find all edges in $L^i_j$.
\end{proof}
\begin{lemma} \label{lemma:correctness}
  Let $E_\ell$ be the set of $F$-light edges in $G$. Let $L = \cup_i \cup_j L^i_j$. Then, $E_\ell = L$. 
\end{lemma}
\begin{proof}
  We first show that $L \subseteq E_\ell$. 
  Consider a Phase $i$ and a component-MWOE pair $(C^i_j, e^i_j)$.
  Consider any edge $e = \{u, v\} \in L^i_j$ with $u \in C^i_j, v\notin C^i_j$. 
  Since $e^i_j$ is the MWOE from $C^i_j$ and $u \in C^i_j$, any path in $F$ connecting $u$ to any node $x \notin C^i_j$ has to go through edge $e^i_j$. 
  Therefore, for any $x \notin C^i_j, w_F(u, x) \geq w(e^i_j)$. 
  Since $v \notin C^i_j$ we have $w_F(u, v) \geq w(e^i_j)$. 
  Moreover, since $e \in L^i_j$, we have $w(e) \leq w(e^i_j)$ implies $w(e) \leq w_F(u, v)$.
  Hence, $e$ is $F$-light.
  Since this is true for any $e \in L^i_j$, we have $L^i_j \subseteq E_\ell$.  
  Hence, $L \subseteq E_\ell$.  

  Now, we show that $E_\ell \subseteq L$. 
  For any node $u \in V$, let $C^q(u)$ denote the component containing $u$ just before
  Phase $q$ of \bor algorithm (Step 2 in Algorithm \textsc{Compute-F-Light}). 
  For the sake of contradiction, let there be an edge $e = \{u, v\} \in E_\ell \setminus L$.  
  Let $i$ be the index of the phase in which component of $u$ and component of $v$ is merged together\footnote{If $u$ and $v$ are never merged
  into one component, i.e., they are in different components in $F$ then $\{u, v\} \in L^i_j$ where $i$ is the phase in which $u$'s component 
  becomes maximal with respect to $F$ and $j$ is such that $u$ belongs to $C^i_j$. This follows from the fact that $e^i_j = \bot$ and $w(e^i_j) = \infty$.} 
  (that is, for any $q < i+1$, $C^q(u) \neq C^q(v)$ and $C^{i+1}(u) = C^{i+1}(v)$). 
  Consider the path $F(u, v)$ and note that since $C^{i+1}(u) = C^{i+1}(v)$, the entire path $F(u, v)$ is in $C^{i+1}(u)$. 
  Now consider the Phase $i$ components $C^i_1, \ldots, C^i_t$, $t \geq 2$ along this path $F(u, v)$ (see Figure~\ref{fig:correctness}). 
  WLOG, let $u \in C^i_1$ and $v \in C^i_t$ and suppose that the path $F(u, v)$ visits the components in the order 
  $u \in C^i_1, C^i_2, \ldots, C^i_{t-1}, v \in C^i_t$. 
  For example, in Figure~\ref{fig:correctness} the path $F(u, v)$ starts in $C^i_1$ then goes through $C^i_2$, then to $C^i_3$, and finally to $C^i_4$. 
  Let $F'(u, v)$ denote the subset of edges in $F(u, v)$ that have endpoints in two distinct Phase $i$ components.

  Now consider the MWOE's of these components: $e^i_j$ is the MWOE for $C^i_j$ for $j = 1, 2, \ldots, t$. 
  There are three cases depending on how the MWOEs $e^i_j$ relate to the path $F(u, v)$.
  \begin{itemize}
  \item $e^i_j$ connects $C^i_j$ to $C^i_{j+1}$ for $j = 1, 2, \ldots, t-1$.
  Since $e$ has exactly one endpoint in  $C^i_1$ and $e \notin L^i_1$ (since $e\notin L$), we have $w(e) > w(e^i_1)$.
  Furthermore, due to the structure of the MWOEs: $w(e^i_1) > w(e^i_2) > \cdots > w(e^i_{t-1})$.
  This implies that $w(e)$ is larger than the weights of all edges in $F'(u, v)$.

  \item $e^i_j$ connects $C^i_j$ to $C^i_{j-1}$ for $j = 2, \ldots, t$.
  Since $e$ has exactly one endpoint in $C^i_t$ and $e \notin L^i_t$ (since $e\notin L$), we have $w(e) > w(e^i_t)$. 
  Furthermore, due to the structure of the MWOEs: $w(e^i_t) > w(e^i_{t-1}) > \cdots > w(e^i_2)$.
  This implies that $w(e)$ is larger than the weights of all edges in $F'(u, v)$.

  \item There is some $\ell$, $1 \le \ell < t$ such that
  $e^i_j$ connects $C^i_j$ to $C^i_{j+1}$ for $j = 1, 2, \ldots, \ell$ and
  $e^i_j$ connects $C^i_j$ to $C^i_{j-1}$ for $j = \ell+1, \ldots, t$.
  This case is illustrated in Figure~\ref{fig:correctness} with $\ell = 2$.
  In this case, $w(e) > w(e^i_1)$ and $w(e) > w(e^i_t)$ for reasons mentioned in the previous two cases.
  Furthermore, due to the structure of the MWOEs: $w(e^i_1) > w(e^i_2) > \cdots > w(e^i_{\ell})$ and
  $w(e^i_t) > w(e^i_{t-1}) > \cdots > w(e^i_{\ell+1})$.
  This implies that $w(e)$ is larger than the weights of all edges in $F'(u, v)$.
  \end{itemize}
  Thus in all three cases, $w(e)$ is larger than the weights of all edges in $F'(u, v)$. 
  Now let $e_F =\{u', v'\} \in F$ be the maximum weight edge in $F(u, v)$.
  Since $e$ is $F$-light, we have $w(e) < w(e_F)$. 
  This inequality combined with the fact that $w(e)$ is larger than the weights of all edges in $F'(u, v)$ implies
  that $u'$ and $v'$ belong to the same Phase $i$ component, i.e., $C^i(u') = C^i(v')$. 
  For example, in Figure~\ref{fig:correctness}, $u'$ and $v'$ are in $C^i_2$. 
  
  Let $C^i(u') = C^i(v') = C^i_\ell$ for some $\ell \leq t$. 
  Let $F(u, v) = F(u, u') \cup \{u', v'\} \cup F(v', v)$. 
  Since $e_F$ is the heaviest edge in $F(u, v)$, all the edges in $F(u, u')$ are lighter than $e_F$. 
  Hence at any Phase $i' < i$, \bor algorithm considers edges in $F(u, u')$ for component $C^{i'}(u')$ and edges in $F(v', v)$ for component $C^{i'}(v')$ before considering $e_F$. 
  The implication of this is, $C^{i}(u) = C^{i}(u')$ and $C^{i}(v) = C^{i}(v')$. 
  But, $C^{i}(u) \neq C^{i}(v)$ therefore, $C^{i}(u') \neq C^{i}(v')$  -- a contradiction. 
\end{proof}
\begin{figure}[tb]
  \centering
  \includegraphics[scale=0.5]{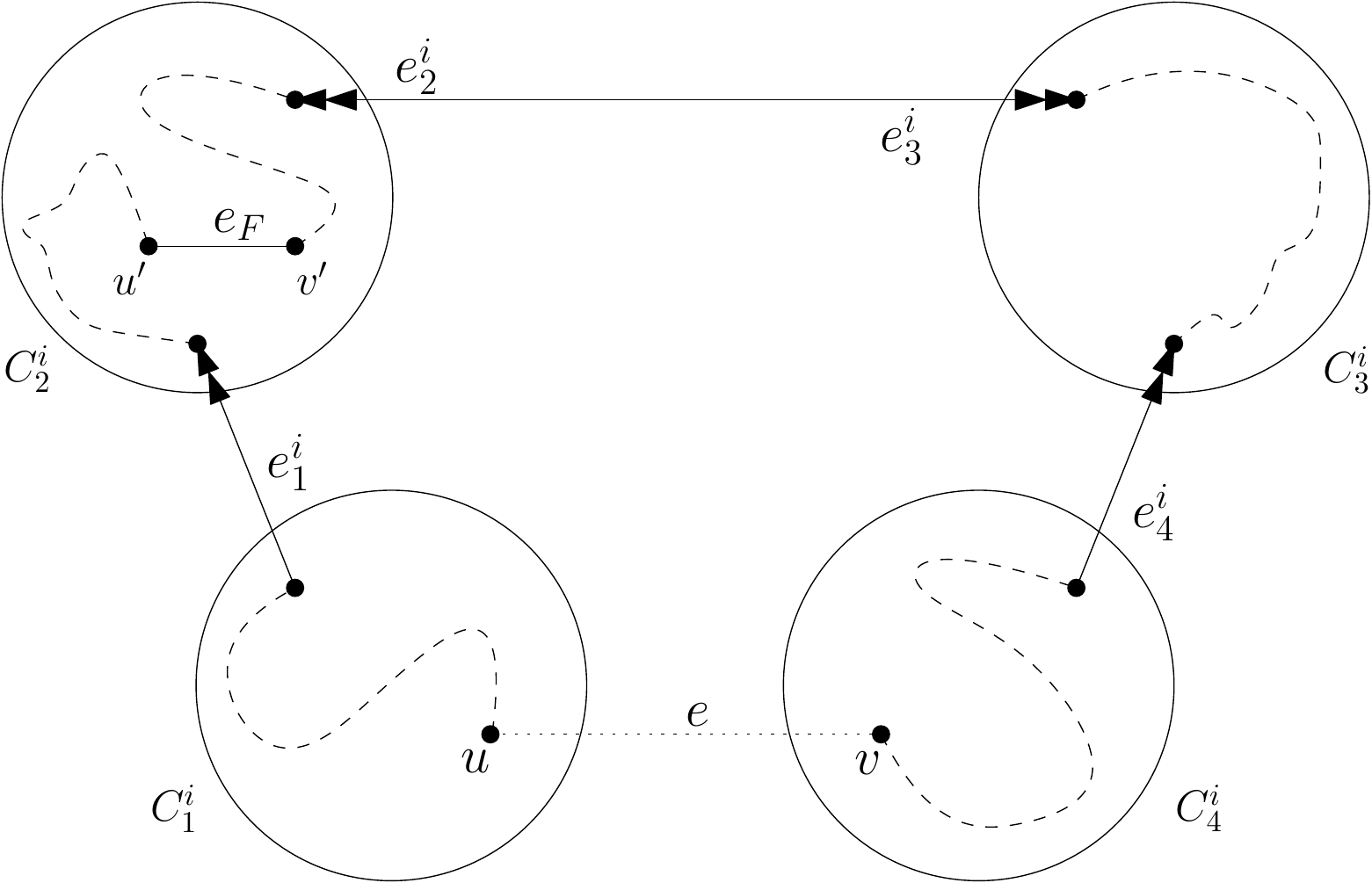}
  \caption{Illustration of proof of Lemma~\ref{lemma:correctness}. 
    After Phase $i$, components $C^i_1, C^i_2, C^i_3, C^i_4$ are merged together using edges $e^i_1, e^i_2, e^i_3, e^i_4$ in $F$. 
    Dashed curves represent paths in $F$ between the respective end-points. 
    $e$ is an $F$-light edge. 
    $e_F$ is the heaviest edge on path from $u$ to $v$ in $F$.\label{fig:correctness}}
\end{figure}
\begin{table}[t]
  \caption{Time and message complexity for steps in Algorithm~\ref{algo:f-light} \textsc{Compute-F-Light}\label{tab:f-light}}
  \centering
  \small
  \begin{tabular}{l l l l}
    \hline
    \textbf{Step} & \textbf{Time} & \textbf{Messages} & \textbf{Analysis} \\ \hline \noalign{\vskip 0.2ex} 
    1 & $O(1)$ & $\ot(n)$ & Theorem~\ref{thm:dsg} \\  
    2 & - & - & Local computation \\
    3 & $O(1)$ & $\ot(n)$ & Trivial direct communication \\
    4 & $O(1)$ & $\ot(n / p)$ & Theorem~\ref{thm:sketches}\\
    5 & $O(1)$ & $\ot(n / p)$ & Lemma~\ref{lemma:step5}\\ \hline
  \end{tabular}
\end{table}
From Lemma~\ref{lemma:lsize} and Lemma~\ref{lemma:correctness} we get the following bound on the number of $F$-light edges in $G$.
\begin{corollary}\label{coro:fsize}
  W.h.p., the number of $F$-light edges in $G$ is $\ot\left({n}/{p}\right)$. 
\end{corollary}
Table~\ref{tab:f-light} summarizes the time and message complexity of each step of Algorithm \textsc{Compute-F-Light}. 
A naive implementation of Step 5 may require super-constant number of rounds because of receiver-side bottlenecks, but we describe here a more sophisticated implementation which runs in $O(1)$ rounds, using $\ot(n/p)$ messages.
\begin{lemma}\label{lemma:step5}
  Step~5 of Algorithm~\ref{algo:f-light} can be implemented in $O(1)$ rounds using $\ot(n / p)$ messages.
\end{lemma}
\begin{proof}
A component $C^i_j$ can be quite large and as a result, the volume of sketches of all nodes in $C^i_j$ can be
much larger than can be received by $C^i_j$'s component leader in $O(1)$ rounds. So before we can gather sketches
at component leaders, we perform two tasks: 
\begin{itemize}
\item[(i)] each commander $v_i$ sets up a simple \textit{rooted tree} communication
structure for each component $C^i_j$, $j = 1, 2, \ldots$ and 
\item[(ii)] $v_i$ informs each node in each component $C^i_j$, $j = 1, 2, \ldots$,
the identity of that node's parent in the rooted tree communication structure.
\end{itemize}
We will show that once these two tasks are completed,
then all requisite sketches can then be gathered at component leaders in $O(1)$ rounds. Of course, we will also
need to show that these two tasks can be completed in $O(1)$ rounds.

Recall that each commander $v_i$ knows $F$ and locally simulates \bor\ algorithm on $F$ and therefore knows the
components $C^i_j$ for all $j$. 
We will now describe how $v_i$ sets up the rooted tree communication structure for a particular component $C^i_j$.
Let $s := \frac{n^{2/3} \cdot p}{\log^9 n}$ and let $S_0 := C^i_j$.
Since $p = \sqrt{n/m}$, we know that $p$ is bounded below by $1/\sqrt{n}$ and therefore
$s \ge \frac{n^{1/6}}{\log^9 n}$.
This shows that $s$ is asymptotically greater than 1 and for the rest of the proof we assume that $s > 1$.
Now commander $v_i$ partitions $S_0$ into $\lceil |S_0|/s \rceil$ subsets, each of size at most $s$.
For each of the $\lceil |S_0|/s \rceil$ parts, node $v_i$ appoints a 
\textit{part leader} (e.g., node with smallest
ID in that part). Let $S_1$ be the set of part leaders. Note that $|S_1| = \lceil |S_0|/s \rceil$.
Next, commander $v_i$ appoints each part leader as the \textit{parent} of all other nodes in that part.

Now $v_i$ repeats this process on $S_1$ to construct the set $S_2$.
In other words, $v_i$ partitions $S_1$ in $\lceil |S_1|/s \rceil$ subsets, each of size at most $s$,
picks part leaders for each of the parts of $S_1$ ($S_2$ is the set of these part leaders), and 
appoints each part leader the parent of all other nodes in its part.
Commander $v_i$ continues in this manner until it generates a set $S_t$ such that $|S_t| \le s$.
Commander $v_i$ then picks a leader for $S_t$ and makes it the parent of all other nodes in $S_t$.
We let $S_{t+1}$ denote the singleton set containing this final leader.
It it easy to see that the choices of part leaders can be made such that the single node
in $S_{t+1}$ is the component leader of $C^i_j$.

Now note that $|S_{i+1}| = \lceil |S_i|/s \rceil$ for $i = 0, 1, \ldots, t$.
Since $|S_0| \le n$ and $s \ge \frac{n^{1/6}}{\log^9 n}$, it follows that
$t = O(1)$ and therefore the rooted tree communication structure we create for each component 
has $O(1)$ depth.

Since each part has size at most $s$, each node in the rooted tree has at most $s$ children.
Now consider how the sketches are sent up this rooted tree to the component leader of $C^i_j$.
First, nodes in $S_0$ are required to send $\Theta\left(\frac{\log^5 n}{p}\right)$ sketches each to 
their parents, i.e., nodes in $S_1$.
Thus each node in $S_1$ needs to receive a total of at most
$$s \times \Theta\left(\frac{\log^5 n}{p}\right) \times \Theta(\log^4 n)  = 
\frac{n^{2/3} \cdot p}{\log^9 n} \times \Theta\left(\frac{\log^5 n}{p}\right) \times \Theta(\log^4 n)  = 
\Theta(n^{2/3})$$
bits. 
This means that we can use the RSG scheme (Theorem \ref{thm:rsg}) to deliver all sketches from nodes in 
$S_0$ to nodes in $S_1$ in $O(1)$ rounds, while keeping the number of messages bounded above by
$O\left(|C^i_j| \cdot \frac{\log^5 n}{p}\right)$.
Once sketches are delivered to nodes in $S_1$, these nodes will aggregate the sketches.
More specifically, suppose that each node $v$ in $S_0$ organizes the $\Theta\left(\frac{\log^5 n}{p}\right)$ sketches
that it sends to its parent, as a vector $(s_1(v), s_2(v), \ldots, s_\beta(v))$ where $\beta = \Theta\left(\frac{\log^5 n}{p}\right)$. 
Each node $w$ in $S_1$, on receiving sketch-vectors from children, computes the following size-$\beta$ vector:
$$\left(\sum_{v} s_1(v), \sum_{v} s_2(v), \ldots, \sum_{v} s_\beta(v)\right).$$
Each of the sums above are over all children $v$ of $w$ (in the rooted tree).
Note that the linearity property of the sketches permits this type of aggregation.
At the end of this step, nodes in $S_1$ have a size-$\beta$ vectors to send to their parents (i.e., nodes in $S_2$).
The above-described process that delivers information from $S_0$ to $S_1$ can be used to deliver information
from $S_1$ to $S_2$, also in $O(1)$ rounds, using $O\left(|C^i_j| \cdot \frac{\log^5 n}{p}\right)$ messages.
Thus, this scheme delivers $\beta = \Theta\left(\frac{\log^5 n}{p}\right)$ component sketches to the component leader
of $C^i_j$ in $O(1)$ rounds while using $O\left(|C^i_j| \cdot \frac{\log^5 n}{p}\right)$ messages.

The routing scheme we have described above can be executed in parallel for all components in 
a particular phase, i.e., for $C^i_j$ for a fixed $i$ and all possible $j$.
We now point out that a stronger claim is true: the above-mentioned routing can be accomplished in parallel for 
all phases as well.
This is because there are $O(\log n)$ phases and thus each node has $O(\log n)$ times as much information to send and 
receive as before (when we were talking about just one phase) and the constraints of the RSG scheme are still met.
Thus this routing scheme delivers information needed by each component leader to compute 
$\Theta\left(\frac{\log^5 n}{p}\right)$ component sketches, in $O(1)$ rounds using 
$O\left(n \cdot \frac{\log^5 n}{p}\right) = \ot(n/p)$ messages.

Finally, we point out that the information on the routing tree communication structure can be communicated by the
commanders to all nodes in 1 communication round. This is because each commander $v_i$ needs to tell each node $v$ 
the ID of $v$'s parent in the routing tree, of $C^i_j$, where $v$ belongs to $C^i_j$. Thus each commander needs
to send $n$ messages to $n$ distinct nodes.
Also note that there are $O(\log n)$ phases in \bor\ algorithm and therefore each node needs to receive messages
from $O(\log n)$ distinct nodes (commanders).
All this can be done by direct communication in 1 round using $O(n \log n)$ messages.
\end{proof}
From Lemma~\ref{lemma:correctness} and Table~\ref{tab:f-light} we get the following result.
\begin{theorem}\label{thm:computeF}
  Algorithm \textsc{Compute-F-Light} computes all $F$-light edges for given graph $G$ and a minimum spanning forest $F$ of $H$ where $H$ is obtained by sampling each edge in $G$ with probability $p$ using a $\Theta(\log n)$-wise-independent sampler. 
  Moreover, the computation takes $O(1)$ rounds and uses $\ot\left({n}/{p}\right)$ messages w.h.p.
\end{theorem}

\renewcommand{\epsilon}{\mathnormal{\oldepsilon}}
\section{Super-Fast Linear-Message-Complexity MST Algorithms}
\label{sec:linear}
In this section we first describe three low-message-complexity routing subroutines and then 
we describe a low-message-complexity sorting subroutine. 
We show that these subroutines can be applied to any of three known ``super-fast'' Congested Clique MST algorithms~\cite{lotker2005mstJournal, Hegeman15podc, GhaffariParter} to reduce their message complexity to $\ot(m)$ while leaving their time complexity unchanged. 
Specifically, we apply these subroutines to the algorithm of Ghaffari and Parter~\cite{GhaffariParter} to obtain an algorithm, we call \textsc{LinearMessages-MST},
that computes an MST of $m$-edge $n$-node input graph in $O(\log^* n)$ rounds and using $\ot(m)$ messages. 

\subsection{Routing Subroutines}
\label{app:routing}
Many recent Congested Clique algorithms have relied on the deterministic routing protocol
due to Lenzen \cite{Lenzen13} that runs in constant rounds on the Congested Clique.
The specific routing problem, called an \textit{Information Distribution Task},
solved by Lenzen's protocol~\cite{Lenzen13} is the following.
Each node $i \in V$ is given a set of $n' \le n$ messages, each of size $O(\log n)$, $\{m_i^1, m_i^2, \ldots, m_i^{n'}\}$,
with destinations $d(m_i^j) \in V$, $j \in [n']$.
Messages are globally lexicographically ordered by their source $i$, destination $d(m_i^j)$, and $j$.
Each node is also the destination of at most $n$ messages.
Lenzen's routing protocol solves the Information Distribution Task in $O(1)$ rounds.
While this subroutine is extremely useful for designing fast Congested Clique algorithms, 
the number of messages is not a resource it tries to explicitly conserve.
Specifically, Lenzen's routing protocol uses $\Omega(n^{1.5})$ messages, independent of
the number of messages that need to be routed.
We observe that the above-mentioned ``super-fast'' MST algorithms do not require the full power of 
Lenzen's routing protocol. What we present below are $O(1)$-round algorithms for slightly restricted routing
problems that use linear number of messages.
These routing protocols suffice for all the routing needs of our MST algorithms.

\begin{theorem}[Randomized Scatter-Gather (\rsg)]\label{thm:rsg}
There are $k$ messages that need to be delivered and each
node is source of up to $n$ messages and each node is destination of up to $c\cdot n^{1-\epsilon}$ messages, where $\epsilon > 0$ and $c \geq 1$ are constants.  
Then there exists an algorithm that, with probability at least $1 - \frac{1}{n}$, delivers all $k$ messages within $\left\lceil 3c/\epsilon\right\rceil$ rounds using $2k$ messages.
\end{theorem}
\begin{proof}
Each node $v$ distributes messages it needs to send, uniformly at random among all nodes, with the constraint that no node gets more than one message.
Each intermediate node then sends the received messages to the specified destinations. 
If an intermediate node receives several messages intended for the same destination, it sends these one-by-one in separate rounds.
We show that w.h.p. no intermediate node will receive more than $\left\lceil {3c}/{\epsilon} \right\rceil$ messages intended for the same destination and 
hence every intermediate node can deliver all messages to destinations in $\left\lceil {3c}/{\epsilon} \right\rceil$ rounds. 

\noindent Let $M_w$ be the set of messages from all senders intended for $w$ and 
let $r_w = |M_w| \leq c\cdot n^{1-\epsilon}$ be the total number of messages intended for $w$. 
Consider a node $u$. 
Let $X_w(u)$ be the random variable denoting the number of messages intended for $w$,  received by $u$ in the first step.
For $m \in M_w$, let $Y_m(u) \in \{0, 1\}$ indicate if $m$ was sent to $u$ in the first step.
Hence $X_w(u) = \sum_{m\in M_w} Y_m(u)$. 
Since $u$ was chosen uniformly at random as the intermediate destination for messages intended to $w$, we have $\E[X_w(u)] \leq \frac{cn^{1-\epsilon}}{n} = c\cdot n^{-\epsilon}$.
Notice that if for any subset of messages in $M_w$ if the sources of these messages is different then the corresponding indicator variables  are independent. 
On the other hand if the source of these messages is the same then they are negatively correlated~\cite{DubhashiBook}. 
Therefore by Chernoff's bound~\cite{DubhashiBook} we have, $\Pr(X_w(u) > c') \leq n^{-2}$ where $c' \leq \left\lceil {3c}/{\epsilon}\right\rceil$.
By the union bound, with probability at least $1-n^{-1}$, each intermediate node will receive at most $\left\lceil {3c}/{\epsilon}\right\rceil$ messages intended for each node and hence can be delivered in less than $\lceil {3c}/{\epsilon}\rceil$ rounds. 
\end{proof}

\noindent By using techniques from~\cite{berns2012facloc, DolevLP12}, we obtain the following result for a particular case of the routing problem. 
\begin{theorem}[Deterministic Scatter-Gather (\dsg)]\label{thm:dsg}
A subset of nodes hold $k$ messages intended for a node $v^*$.
Then there exists a deterministic algorithm that delivers all $k$ messages within $2\left\lceil {k}/{n} \right\rceil + 2$ rounds using $2k + 2$ messages.
Moreover, this can be extended to a scenario where 
there is a set $V^* \subseteq V$ of destinations and every message needs to be delivered to every node in $V^*$. 
In this case, the algorithm terminates in $2\left\lceil {k}/{n} \right\rceil + 2$ rounds using $(2k + 2)|V^*|$ messages. 
\end{theorem}
Now consider the reverse scenario:
\begin{theorem}[Deterministic Gather-Scatter (\dgs)]\label{thm:dgs} 
A node $v^*$ holds a bulk of messages intended for a subset of nodes $R\subseteq V$ such that the total number of messages is $k \leq n$ and each message needs to delivered to \emph{all} nodes in $R$. 
Then there exists a deterministic algorithm that delivers all $k$ messages within $2$ rounds using $k + k\cdot|R|$ messages. 
\end{theorem}
\begin{proof}
Node $v^*$ sends each message $m_i$ to a \emph{supporter node} $s_i$. 
Since $k < n$, an one-to-one mapping of $m_i$ to $s_i$ is possible and hence this can be done in a single round and uses $k$ messages.
Each supporter node then broadcast the received message to all nodes in $R$.
This requires one round and $k\cdot |R|$ messages. 
\end{proof}

\subsection{Sorting Subroutine}
\label{app:linearMST}

The Ghaffari and Parter MST algorithm (\textsc{GP-MST}) is partly based on techniques of Hegeman et al.~\cite{Hegeman15podc} and one of the key ideas there is to sort edges in the input graph based on weights. 
\textsc{GP-MST} and Hegeman et al.~\cite{Hegeman15podc} both rely on the $O(1)$-round deterministic sorting routine by Lenzen~\cite{Lenzen13} which requires $\Omega(n^{1.5})$ messages regardless of the number of keys to sort. 
In addition to the low-message-complexity routing primitives mentioned above, we develop a new low-message-complexity 
sorting primitive (based on the Congested Clique sorting algorithm of \cite{Lenzen13}).

Consider the following problem: 
given $k$ keys of size $O(\log n)$ each from a totally ordered universe such that each node has up to $n$ keys. 
The goal is to learn the rank of each of these keys in a global ordered enumeration of all $k$ keys, i.e.,
each node should learn the ranks of the keys it is holding. 
Patt-Shamir and Teplitsky~\cite{Patt-ShamirT11} designed a randomized algorithm that solved 
this problem in $O(\log \log n)$ rounds which was later improved to $O(1)$ rounds by the deterministic algorithm of Lenzen~\cite{Lenzen13}. 
But, both the algorithms~\cite{Patt-ShamirT11, Lenzen13} have $\Omega(n^{1.5})$ message complexity regardless of the number of keys to sort.
We provide a randomized algorithm which reduces the problem to the similar problem as above but on a smaller clique. 
Our algorithm solves the problem for $k = O(n^{2-\epsilon}), \epsilon > 0$ in $O(1)$ rounds using $O(k)$ messages w.h.p.

The high level idea of our Algorithm \textsc{DistributedSort} is to redistribute $k$ keys to $\lfloor \sqrt{k} \rfloor$ nodes and then sort them using Lenzen's sorting algorithm~\cite{Lenzen13} on the clique induced by these $\lfloor \sqrt{k} \rfloor$ nodes in $O(1)$ rounds with $O(k)$ messages. 
For the redistribution, we rely on our low-message routing schemes (\rsg~and \dsg). 
Let $k_v$ be the number of keys $v$ has. 
Each node $v$ sends $k_v$ to node $v^*$.
Notice that, $k = \sum_{w \in V} k_w$. 
Let $idx_w = \sum_{u : ID(u) < ID(w)} k_u$ for all $w \in V$.
For each $w \in V$, $v^*$ sends $idx_w$ to $w$ .
Order keys present at each node $v$ arbitrarily. 
Assign labels to keys starting from $idx_v$. 
Set destination of the key with label $i$ to node $\left(i\mod \lfloor \sqrt{k} \rfloor\right)$. 
At this point the input is divided among $\lfloor \sqrt{k} \rfloor$ nodes, each holding up to $\lceil \sqrt{k} \rceil$ keys. 
Let $V_\mu$ denote the set of nodes with IDs in the range $[0, \lfloor \sqrt{k} \rfloor-1]$.
Nodes in $V_\mu$ executes Lenzen's sorting algorithm~\cite{Lenzen13} and learn the global index of the keys in sorted order.
Each key with its rank in global sorted order is sent back to the original node (by reversing the route applied earlier to this key). 

\begin{theorem}[Distributed Sorting]\label{thm:sorting}
  Given $k = O(n^{2 - \epsilon})$ comparable keys of size $O(\log n)$ each such that each node has up to $n$ keys for some constant $\epsilon > 0$.
  Then, Algorithm \textsc{DistributedSort} requires $O\left({1}/{\epsilon}\right)$ rounds and $O(k)$ messages w.h.p., 
  such that at the end of the execution each node knows the rank of each key it has. 
\end{theorem}
\begin{proof}
  We first show that the redistribution of keys among $\lfloor \sqrt{k} \rfloor$ nodes takes $O(1)$ rounds and $O(k)$ messages. 
  Since each of the $\lfloor \sqrt{k} \rfloor$ nodes need to receive $\lceil \sqrt{k} \rceil = O(n^{1-\epsilon})$ keys, the keys can be routed using the \rsg~(Theorem~\ref{thm:rsg}) in $O(1)$ rounds and $O(k)$ messages. 
  Nodes in $V_\mu$ can now execute Lenzen's sorting algorithm~\cite{Lenzen13} which takes $O(1)$ rounds and $O(k)$ messages. 
  The reverse routing of these keys takes another $O(1)$ rounds and $O(k)$ messages. 
  Therefore, in total Algorithm \textsc{DistributedSort} required $O(1)$ rounds and $O(k)$ messages. 
\end{proof}

We obtain the following result by replacing the routing and sorting routines due to Lenzen~\cite{Lenzen13} used in \textsc{GP-MST} with our routing and sorting routines developed above.
\begin{theorem}[\textsc{LinearMessages-MST}]\label{thm:linearMST}
  There exist a MST algorithm that computes a minimum spanning tree of an $n$-node $m$-edge input graph in $O(\log^* n)$ rounds using $\ot(m)$ messages w.h.p. in the Congested Clique. 
\end{theorem}

\bibliography{mymst}

\begin{thebibliography}{10}

\bibitem{AhnSODA12}
Kook~Jin Ahn, Sudipto Guha, and Andrew McGregor.
\newblock {Analyzing graph structure via linear measurements}.
\newblock In {\em {Proceedings of the 23rd annual ACM-SIAM Symposium on
  Discrete Algorithms (SODA)}}, pages 459--467, 2012.

\bibitem{AhnPODS12}
Kook~Jin Ahn, Sudipto Guha, and Andrew McGregor.
\newblock {Graph sketches: sparsification, spanners, and subgraphs}.
\newblock In {\em {Proceedings of the 31st Symposium on Principles of Database
  Systems (PODS)}}, pages 5--14, 2012.

\bibitem{berns2012facloc}
Andrew Berns, James Hegeman, and Sriram~V. Pemmaraju.
\newblock {Super-Fast Distributed Algorithms for Metric Facility Location}.
\newblock In {\em {Proccedings of the 39th International Colloquium on
  Automata, Languages, and Programming (ICALP)}}, pages 428--439, 2012.

\bibitem{Censor15podc}
Keren Censor{-}Hillel, Petteri Kaski, Janne~H. Korhonen, Christoph Lenzen, Ami
  Paz, and Jukka Suomela.
\newblock Algebraic methods in the congested clique.
\newblock In Chryssis Georgiou and Paul~G. Spirakis, editors, {\em Proceedings
  of the 2015 {ACM} Symposium on Principles of Distributed Computing, {PODC}
  2015, Donostia-San Sebasti{\'{a}}n, Spain, July 21 - 23, 2015}, pages
  143--152. {ACM}, 2015.
\newblock URL: \url{http://doi.acm.org/10.1145/2767386.2767414}.

\bibitem{cormode2014sampling}
Graham Cormode and Donatella Firmani.
\newblock {A unifying framework for $\ell_0$-sampling algorithms}.
\newblock {\em Distributed and Parallel Databases}, 32(3):315--335, 2014.

\bibitem{DolevLP12}
Danny Dolev, Christoph Lenzen, and Shir Peled.
\newblock {``{T}ri, Tri Again'': Finding Triangles and Small Subgraphs in a
  Distributed Setting}.
\newblock In {\em {Proceedings of the 26th International Symposium on
  Distributed Computing (DISC)}}, pages 195--209, 2012.

\bibitem{drucker2012task}
Andrew Drucker, Fabian Kuhn, and Rotem Oshman.
\newblock {The communication complexity of distributed task allocation.}
\newblock In {\em {Proceedings of the 30st ACM Symposium on Principles of
  Distributed Computing (PODC)}}, pages 67--76, 2012.
\newblock URL: \url{http://doi.acm.org/10.1145/2332432.2332443}.

\bibitem{DubhashiBook}
Devdatt~P. Dubhashi and Alessandro Panconesi.
\newblock {\em {Concentration of Measure for the Analysis of Randomized
  Algorithms}}.
\newblock Cambridge University Press, 2009.

\bibitem{Elkin2006}
Michael Elkin.
\newblock {An Unconditional Lower Bound on the Time-Approximation Trade-off for
  the Distributed Minimum Spanning Tree Problem}.
\newblock {\em {SIAM} J. Comput.}, 36(2):433--456, 2006.
\newblock URL: \url{http://dx.doi.org/10.1137/S0097539704441058}, \href
  {http://dx.doi.org/10.1137/S0097539704441058}
  {\path{doi:10.1137/S0097539704441058}}.

\bibitem{GehweilerSPAA2006}
Joachim Gehweiler, Christiane Lammersen, and Christian Sohler.
\newblock {A Distributed ${O}(1)$-approximation Algorithm for the Uniform
  Facility Location Problem}.
\newblock In {\em {Proceedings of the 18th Annual ACM Symposium on Parallelism
  in Algorithms and Architectures (SPAA)}}, pages 237--243, 2006.
\newblock URL: \url{http://doi.acm.org/10.1145/1148109.1148152}.

\bibitem{GhaffariParter}
Mohsen Ghaffari and Merav Parter.
\newblock {MST in Log-Star Rounds of Congested Clique}.
\newblock In {\em {Proceedings of the 2016 ACM Symposium on Principles of
  Distributed Computing}}, {PODC '16}, 2016.

\bibitem{Hegeman15podc}
James~W. Hegeman, Gopal Pandurangan, Sriram~V. Pemmaraju, Vivek~B. Sardeshmukh,
  and Michele Scquizzato.
\newblock {Toward Optimal Bounds in the Congested Clique: Graph Connectivity
  and MST}.
\newblock In {\em {Proceedings of the 2015 ACM Symposium on Principles of
  Distributed Computing}}, {PODC '15}, pages 91--100. ACM, 2015.
\newblock URL: \url{http://doi.acm.org/10.1145/2767386.2767434}.

\bibitem{HegemanP14}
James~W. Hegeman and Sriram~V. Pemmaraju.
\newblock {Lessons from the Congested Clique Applied to {MapReduce}}.
\newblock In {\em {Proceedings of the 21th International Colloquium on
  Structural Information and Communication Complexity (SIROCCO)}}, pages
  149--164, 2014.

\bibitem{Hegeman2014disc}
James~W. Hegeman, Sriram~V. Pemmaraju, and Vivek~B. Sardeshmukh.
\newblock {Near-Constant-Time Distributed Algorithms on a Congested Clique}.
\newblock In {\em {Proceedings of the 28th International Symposium on
  Distributed Computing (DISC)}}, pages 514--530, 2014.

\bibitem{HolzerP14}
Stephan Holzer and Nathan Pinsker.
\newblock {Approximation of Distances and Shortest Paths in the Broadcast
  Congest Clique}.
\newblock {\em CoRR}, abs/1412.3445, 2014.

\bibitem{JowhariL0}
Hossein Jowhari, Mert Sa\u{g}lam, and G\'{a}bor Tardos.
\newblock {Tight Bounds for Lp Samplers, Finding Duplicates in Streams, and
  Related Problems}.
\newblock In {\em {Proceedings of the Thirtieth ACM SIGMOD-SIGACT-SIGART
  Symposium on Principles of Database Systems}}, {PODS '11}, pages 49--58. ACM,
  2011.
\newblock URL: \url{http://doi.acm.org/10.1145/1989284.1989289}, \href
  {http://dx.doi.org/10.1145/1989284.1989289}
  {\path{doi:10.1145/1989284.1989289}}.

\bibitem{KKT1995MST}
David~R. Karger, Philip~N. Klein, and Robert~E. Tarjan.
\newblock {A Randomized Linear-time Algorithm to Find Minimum Spanning Trees}.
\newblock {\em J. ACM}, 42(2):321--328, March 1995.
\newblock URL: \url{http://doi.acm.org/10.1145/201019.201022}, \href
  {http://dx.doi.org/10.1145/201019.201022} {\path{doi:10.1145/201019.201022}}.

\bibitem{king15testout}
Valerie King, Shay Kutten, and Mikkel Thorup.
\newblock Construction and impromptu repair of an mst in a distributed network
  with o(m) communication.
\newblock In {\em Proceedings of the 2015 ACM Symposium on Principles of
  Distributed Computing}, PODC '15, pages 71--80, New York, NY, USA, 2015. ACM.
\newblock URL: \url{http://doi.acm.org/10.1145/2767386.2767405}, \href
  {http://dx.doi.org/10.1145/2767386.2767405}
  {\path{doi:10.1145/2767386.2767405}}.

\bibitem{KlauckNPR15}
Hartmut Klauck, Danupon Nanongkai, Gopal Pandurangan, and Peter Robinson.
\newblock {Distributed Computation of Large-Scale Graph Problems}.
\newblock In {\em {Proceedings of the 26th Annual {ACM-SIAM} Symposium on
  Discrete Algorithms (SODA)}}, pages 391--410, 2015.

\bibitem{KorhonenArxiv16}
Janne~H. Korhonen.
\newblock Deterministic {MST} sparsification in the congested clique.
\newblock {\em CoRR}, abs/1605.02022, 2016.
\newblock URL: \url{http://arxiv.org/abs/1605.02022}.

\bibitem{KuttenPPRTJACM2015}
Shay Kutten, Gopal Pandurangan, David Peleg, Peter Robinson, and Amitabh
  Trehan.
\newblock On the complexity of universal leader election.
\newblock {\em J. ACM}, 62(1):7:1--7:27, March 2015.
\newblock URL: \url{http://doi.acm.org/10.1145/2699440}, \href
  {http://dx.doi.org/10.1145/2699440} {\path{doi:10.1145/2699440}}.

\bibitem{KuttenPeleg1998}
Shay Kutten and David Peleg.
\newblock {Fast Distributed Construction of Small \emph{k}-Dominating Sets and
  Applications}.
\newblock {\em J. Algorithms}, 28(1):40--66, 1998.
\newblock URL: \url{http://dx.doi.org/10.1006/jagm.1998.0929}, \href
  {http://dx.doi.org/10.1006/jagm.1998.0929}
  {\path{doi:10.1006/jagm.1998.0929}}.

\bibitem{Lenzen13}
Christoph Lenzen.
\newblock {Optimal Deterministic Routing and Sorting on the Congested Clique}.
\newblock In {\em {Proceedings of the 31st {ACM} Symposium on Principles of
  Distributed Computing (PODC)}}, pages 42--50. ACM, 2013.
\newblock URL: \url{http://doi.acm.org/10.1145/2484239.2501983}, \href
  {http://dx.doi.org/10.1145/2484239.2501983}
  {\path{doi:10.1145/2484239.2501983}}.

\bibitem{lotker2005mstJournal}
Zvi Lotker, Boaz Patt{-}Shamir, Elan Pavlov, and David Peleg.
\newblock {Minimum-Weight Spanning Tree Construction in ${O}(\log \log n)$
  Communication Rounds}.
\newblock {\em SIAM Journal on Computing}, 35(1):120--131, 2005.

\bibitem{McGregorSurvey}
Andrew McGregor.
\newblock {Graph stream algorithms: A survey}.
\newblock {\em ACM SIGMOD Record}, 43(1):9--20, 2014.

\bibitem{MotwaniRaghavan}
Rajeev Motwani and Prabhakar Raghavan.
\newblock {\em {Randomized Algorithms}}.
\newblock Cambridge University Press, New York, NY, USA, 1995.

\bibitem{Nanongkai14}
Danupon Nanongkai.
\newblock {Distributed approximation algorithms for weighted shortest paths}.
\newblock In {\em {Proceedings of the 46th {ACM} Symposium on Theory of
  Computing (STOC)}}, pages 565--573, 2014.

\bibitem{PanduranganRS16}
Gopal Pandurangan, Peter Robinson, and Michele Scquizzato.
\newblock A time- and message-optimal distributed algorithm for minimum
  spanning trees.
\newblock {\em CoRR}, abs/1607.06883, 2016.
\newblock URL: \url{http://arxiv.org/abs/1607.06883}.

\bibitem{Patt-ShamirT11}
Boaz Patt-Shamir and Marat Teplitsky.
\newblock {The Round Complexity of Distributed Sorting}.
\newblock In {\em {Proceedings of the 30th Annual ACM Symposium on Principles
  of Distributed Computing (PODC)}}, pages 249--256, 2011.
\newblock URL: \url{http://doi.acm.org/10.1145/1993806.1993851}, \href
  {http://dx.doi.org/10.1145/1993806.1993851}
  {\path{doi:10.1145/1993806.1993851}}.

\bibitem{peleg2000distributed}
David Peleg.
\newblock {\em {Distributed Computing: A Locality-Sensitive Approach}}.
\newblock Society for Industrial Mathematics, 2000.

\bibitem{PettieRamachandran}
Seth Pettie and Vijaya Ramachandran.
\newblock {Randomized minimum spanning tree algorithms using exponentially
  fewer random bits}.
\newblock {\em {ACM} Trans. Algorithms}, 4(1), 2008.
\newblock URL: \url{http://doi.acm.org/10.1145/1328911.1328916}, \href
  {http://dx.doi.org/10.1145/1328911.1328916}
  {\path{doi:10.1145/1328911.1328916}}.

\bibitem{DasSarmaSICOMP2011}
Atish~Das Sarma, Stephan Holzer, Liah Kor, Amos Korman, Danupon Nanongkai,
  Gopal Pandurangan, David Peleg, and Roger Wattenhofer.
\newblock {Distributed Verification and Hardness of Distributed Approximation}.
\newblock {\em {SIAM} J. Comput.}, 41(5):1235--1265, 2012.

\bibitem{SchmidtSS95}
Jeanette~P. Schmidt, Alan Siegel, and Srinivasan Aravind.
\newblock {Chernoff-Hoeffding Bounds for Applications with Limited
  Independence}.
\newblock {\em {SIAM} J. Discrete Math.}, 8(2):223--250, 1995.
\newblock URL: \url{http://dx.doi.org/10.1137/S089548019223872X}, \href
  {http://dx.doi.org/10.1137/S089548019223872X}
  {\path{doi:10.1137/S089548019223872X}}.

\bibitem{TarjanBook}
Robert~Endre Tarjan.
\newblock {\em {CBMS-NSF Regional Conference Series in Applied Mathematics:
  Data Structures and Network Algorithms}}.
\newblock Society for Industrial and Applied Mathematics, New York, NY, USA,
  1983.

\end{thebibliography}
\end{document}